\documentclass{ieeeaccess}
\usepackage[latin1]{inputenc}
\usepackage{graphicx}
\usepackage{amssymb,amsmath,amsfonts,amsthm}
\usepackage{float}
\usepackage{graphics}
\usepackage{xspace}
\usepackage{amssymb, epsfig, hyperref}
\usepackage{amsmath}
\usepackage{physics}
\usepackage{epstopdf}
\newtheorem{proposition}{Proposition}

\newtheorem{theorem}{Theorem}

\usepackage{caption}

\def\BibTeX{{\rm B\kern-.05em{\sc i\kern-.025em b}\kern-.08em
    T\kern-.1667em\lower.7ex\hbox{E}\kern-.125emX}}
\begin{document}
\history{Date of publication xxxx 00, 0000, date of current version xxxx 00, 0000.}
\doi{10.1109/TQE.2020.DOI}

\title{Quantum Rotation Diversity in Displaced Squeezed Binary Phase-Shift Keying}
\author{\uppercase{Ioannis Krikidis} \IEEEmembership{Fellow, IEEE}}
\address{Department of Electrical and Computer Engineering, University of Cyprus, Cyprus (email: krikidis@ucy.ac.cy)}
\tfootnote{This work received funding from the European Research Council (ERC) under the European Union's Horizon Europe programme (ERC Proof of Concept QUARTO, Grant Agreement No. 101241675) and from the European Union's Horizon Europe programme (FOCAL project, Grant Agreement No. 101225859).}


\markboth
{I. Krikidis: Quantum Rotation Diversity in Displaced Squeezed Binary Phase-Shift Keying}
{I. Krikidis: Quantum Rotation Diversity in Displaced Squeezed Binary Phase-Shift Keying}

\corresp{Corresponding author: Ioannis Krikidis (email: krikidis@ucy.ac.cy).}

\begin{abstract}
We propose a quantum rotation diversity (QRD) scheme for optical quantum communication using binary phase-shift-keying displaced squeezed states and homodyne detection over Gamma--Gamma turbulence channels. Consecutive temporal modes are coupled by a passive orthogonal rotation that redistributes the displacement amplitude between slots, yielding a diversity order of two under independent fading and joint maximum-likelihood detection. Analytical expressions for the symbol-error rate performance, along with asymptotic results for the diversity and coding gains, are derived. The optimal rotation angle and energy allocation between displacement and squeezing are obtained in closed form. Furthermore, we show that when both the displacement amplitude and the squeezing strength scale with the total photon number, an effective diversity order of four is achieved. Numerical results validate the analysis and demonstrate the super-diversity behaviour of the proposed QRD scheme.
\end{abstract}


\titlepgskip=-15pt

\maketitle

\section{Introduction}

\PARstart{Q}{uantum} optical communication is an emerging technology that aims to unlock non-classical capabilities by exploiting fundamental principles of quantum mechanics, such as entanglement, superposition, and quantum measurement \cite{DJO}. It employs photons (light quanta) to encode classical information into quantum states of light, which are transmitted through optical fibers or free-space optical (FSO) channels. Specifically, quantum FSO communication combines the advantages of quantum optical signal processing with the flexibility of the wireless medium and is relevant to application scenarios such as ground-to-satellite and satellite-to-ground optical links, as well as terrestrial FSO systems operating under strong atmospheric turbulence or stringent power constraints. More broadly, quantum FSO communication constitutes an important physical-layer building block for emerging technologies such as quantum-enhanced sensing and security (quantum key distribution), low-power optical Internet of Things, and hybrid classical-quantum communication networks \cite{DJO2, LAZ}. In these scenarios, deep channel fades and limited received optical power can drive the system into the photon-limited regime, where performance is fundamentally constrained by quantum noise rather than classical thermal noise.

Motivated by these application scenarios, this communication paradigm has recently gained significant attention within communication theory and information processing. For instance, the study in \cite{JUN} investigates coding techniques for wireless quantum optical communications over multipath channels, while the work in \cite{KRI} proposes a batteryless wireless quantum optical link and analyzes its fundamental performance limits. The study in \cite{BHA} examines a wireless quantum optical relay channel over Gamma--Gamma turbulence and evaluates its performance in terms of the Helstrom error bound. Finally, the work in \cite{CHE} investigates how optical squeezing can enhance the discrimination of binary phase-shift-keying (BPSK) quantum signals under noise, demonstrating superior detection performance compared to conventional coherent-state schemes. 

On the other hand, transmit diversity is a fundamental concept in wireless communication theory that enhances the reliability of communication systems under fading conditions. The key idea is to transmit the same information through multiple independent paths, thereby increasing the probability that at least one copy is received reliably. Diversity can be implemented using various techniques, {\it e.g.}, time, frequency, or space diversity \cite{TSE}. The incorporation of this principle into quantum optical communications opens a new research direction toward more efficient and robust quantum optical links, particularly under FSO channel impairments. Existing works rely on the principle of repetition coding, where the same information is transmitted through parallel and independent channels ({\it e.g.}, spatial modes) and combined at the receiver to achieve diversity gains \cite{CHA}. However, these approaches improve reliability at the expense of reduced communication throughput. The design of diversity techniques for quantum optical communications that preserve throughput remains an open research problem in the literature.

Motivated by recent advances in wireless communication theory, this work investigates a novel transmit diversity scheme for wireless quantum optical communications that achieves diversity gains while preserving throughput.
Building upon the concept of signal-space (modulation) diversity \cite{BOU}, which provides robustness against fading via constellation rotation in classical wireless communications, we introduce the quantum rotation diversity (QRD) scheme, where consecutive modulated quantum states are jointly rotated to enable temporal diversity without introducing redundancy. Specifically, the contributions of this paper are twofold:
\begin{itemize}
\item We consider a BPSK quantum optical link employing displaced squeezed states \cite{CHE} and homodyne detection over Gamma--Gamma turbulence channels, which model amplitude fading in FSO links. The proposed QRD scheme applies a passive orthogonal rotation across consecutive temporal modes, redistributing the displacement amplitude over two time slots in a manner analogous to constellation rotation in \cite{BOU}. At the receiver, joint maximum-likelihood (ML) detection is employed across the two time slots to exploit independent fading realizations and enable temporal modulation diversity without introducing redundancy, while fully preserving the transmission rate.
\item An analytical framework is developed for the symbol-error rate (SER) performance, including exact and asymptotic expressions that characterize the achievable diversity and coding gains. Analytical and numerical results demonstrate that the proposed QRD scheme attains a diversity order of two under independent block fading, which is doubled compared to conventional no-rotation schemes. The optimal power allocation between displacement and squeezing is formulated as an optimization problem that minimizes the SER, and closed-form solutions are derived for both the non-asymptotic and asymptotic regimes. Furthermore, we show that the achievable diversity order increases to four when both the displacement amplitude and the squeezing strength asymptotically scale with the photon number, revealing a super-diversity behavior.
\end{itemize}

\noindent {\it Notation:} 
$\mathrm{B}(\cdot)$, $\Gamma(\cdot)$, and $K_{\nu}(\cdot)$ denote the Beta, Gamma, and modified Bessel functions of the second kind, respectively \cite{GRA};  $Q(\cdot)$ denotes the Gaussian $Q$-function. 
Bold lowercase letters $\pmb{x}$ denote column vectors, while bold uppercase letters $\pmb{X}$ denote matrices; $\pmb{I}_n$ represents the identity matrix of dimension $n\times n$, $(\cdot)^{\mathsf{T}}$ denotes the transpose operation, $\mathbb{E}(\cdot)$ denotes the expectation operator, and $\mathcal{N}(\mu,\sigma^2)$ represents the Gaussian probability density function (PDF) with mean $\mu$ and variance $\sigma^2$.

\begin{figure*}
	\centering
	\includegraphics[width=0.8\linewidth]{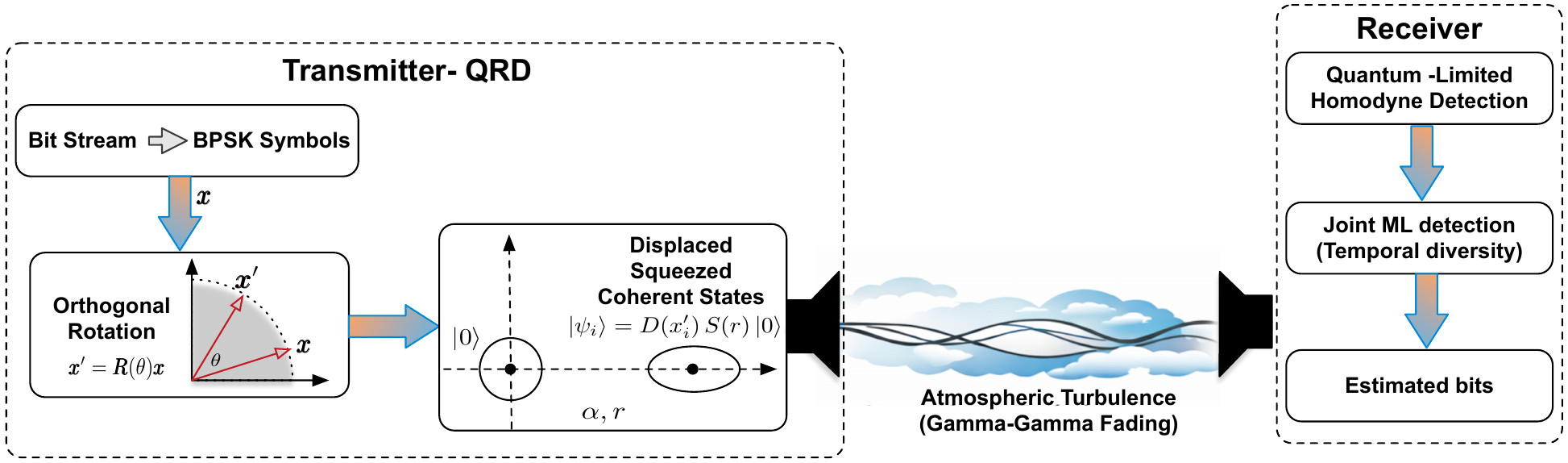}
	\caption{System model of the proposed QRD-based quantum FSO communication scheme.}
	\label{fig0}
\end{figure*}

\section{System model and the proposed QRD scheme}

We consider a binary quantum optical communication link in which the transmitter encodes information using BPSK-modulated displaced squeezed states, and the receiver employs homodyne detection for symbol discrimination. To enhance link reliability over turbulent FSO channels, we propose the QRD scheme, which operates across two consecutive time slots. In this scheme, a rotation matrix is applied to pairs of BPSK symbols at the transmitter, while the receiver performs joint homodyne-based detection over the two slots. This design achieves a diversity order of two, thereby significantly improving robustness against FSO amplitude fading. The diversity gain arises because each information symbol is spread across two independently faded time slots and jointly detected, so that an error occurs only when both fading realizations are simultaneously weak. Specifically, we consider two consecutive BPSK symbols $x_1,x_2\in\{+\alpha,-\alpha\}$ with $\alpha \in \mathbb{R}$, and the transmitter applies a real orthonormal rotation to the symbol pair to introduce modulation diversity across two time slots \cite{BOU} {\it i.e.,}
\begin{equation}
	\pmb{x}'=\pmb{R}(\theta)\pmb{x},
	\;\;\textrm{with}\;\;
	\pmb{R}(\theta)=
	\begin{bmatrix}
		\cos\theta & -\sin\theta \\
		\sin\theta & \phantom{-}\cos\theta
	\end{bmatrix},
	\label{eq:rotation}
\end{equation}
where $\pmb{x}'=[x_1'\; x_2']^\mathsf{T}$, $\pmb{x}=[x_1\;x_2]^\mathsf{T}$, and $\theta \in [0,\,\pi/4]$\footnote{The rotation angle $\theta$ is restricted to the interval $[0,\,\pi/4]$, since the orthogonal rotation matrix exhibits 
	$\pi/2$- periodic symmetry, and rotations beyond this range yield geometrically equivalent constellations with identical performance.}. Therefore, the transmitter emits rotated combinations of the original BPSK symbols,
\begin{equation}
	x'_1=x_1\cos\theta - x_2\sin\theta, \qquad
	x'_2=x_1\sin\theta + x_2\cos\theta.
\end{equation}
The rotation preserves energy ({\it i.e.,}
$|x'_1|^2+|x'_2|^2=|x_1|^2+|x_2|^2$), and spreads the information of each
symbol across both time slots. Each rotated amplitude $x'_i$ is then used to modulate a single-mode displaced  squeezed state at the optical transmitter. The $i$-th transmitted quantum state is given by
\begin{equation}
	\ket{\psi_i} = D(x'_i)\,S(r)\ket{0},
	\label{eq:state}
\end{equation}
where $\ket{0}$ is the vacuum state, $D(\alpha)= \exp\!\left(\alpha\hat{a}^\dagger - \alpha^*\hat{a}\right)$ and $S(r) = \exp\!\left[\tfrac{1}{2}r\!\left(\hat{a}^2-\hat{a}^{\dagger2}\right)\right]$ denote the displacement and squeezing operators, respectively, $r \in \mathbb{R}$, $\hat a$ and $\hat a^\dagger$ denote the annihilation and creation operators, respectively, satisfying the canonical commutation relation
$[\hat{a},\hat{a}^\dagger]=1$ \cite{DJO}. The displacement amplitude $x'_i$ encodes  the BPSK symbol information, while the real squeezing parameter $r$ controls the  quantum noise reduction along the measured quadrature. Let $N=\alpha^2+\sinh^2(r)$ represents the total average photon number available per transmitted symbol, defining the per-symbol energy budget shared between displacement and squeezing ({\it i.e.,} $ \alpha^{2}$ and $\sinh^{2}(r)$ represent the average number of photons contributed by the displacement and squeezing processes, respectively \cite{CHE}).

The optical signal is transmitted through a classical FSO channel 
subject to atmospheric turbulence, which induces random amplitude fluctuations  on the received field \cite{BHA}. The irradiance (or intensity) fading is  as a  Gamma--Gamma distributed random variable, whose PDF is given by
\begin{equation}
	f_I(z) = 
	\frac{2(\epsilon \zeta)^{\frac{\epsilon+\zeta}{2}}}
	{\Gamma(\epsilon)\Gamma(\zeta)}\,
	z^{\frac{\epsilon+\zeta}{2}-1}
	K_{\epsilon-\zeta}\!\left(2\sqrt{\epsilon \zeta z}\right),
	\qquad z>0,
	\label{eq:gg_pdf}
\end{equation}
where the shape parameters $\epsilon,\zeta>0$  with $\epsilon\neq \zeta$ represent the large-scale (beam wander) and small-scale (scintillation) turbulence effects, respectively. Although Gamma--Gamma turbulence is adopted for analytical tractability, the proposed QRD framework is generic and can, in principle, be applied to other fading models in FSO communications.

At the receiver, homodyne detection is employed to measure the quadrature operator of the input field, defined as $\hat{x}_{\phi} = \frac{1}{\sqrt{2}}(\hat{a}e^{-i\phi} + \hat{a}^{\dagger}e^{i\phi}),$ where $\phi$ denotes the phase of the measured quadrature. Since real-valued displacement and squeezing parameters are assumed, the homodyne measurement corresponds to $\phi = 0$, \textit{i.e.}, to the in-phase quadrature of the field. Under this assumption, the output of the homodyne detector at the $i$-th time slot (with $i=1,2$) can be expressed as
\begin{equation}
	y_i = \sqrt{\eta I_i}\,x'_i + n_{q,i},
	\label{eq:homodyne_output}
\end{equation}
where $\eta\in[0,1]$ denotes the normalized effective path loss, $I_i$ are independent and identically distributed (i.i.d) Gamma--Gamma irradiance coefficients, and $n_{q,i}$ represents the
quantum noise\footnote{For simplicity, we consider only the effects of quantum noise \cite{CHE}; additional thermal or background noise would mainly affect the effective noise variance and is not expected to alter the qualitative conclusions or asymptotic trends of the analysis.} term, modelled as a real Gaussian random variable
$n_{q,i}\sim\mathcal{N}\!\left(0,\,\sigma_q^2\right)$ with variance $\sigma_q^2 = \tfrac{1}{2} e^{-2r}$. The received two-dimensional observation vector corresponding to the pair of
time slots is given by
\begin{equation}
	\pmb{y}=\pmb{H}\pmb{x}+\pmb{n}_q,
\end{equation}
where $\pmb{H}=\sqrt{\eta}
\begin{bmatrix}
	\!\sqrt{I_1} & 0\\
	0 & \sqrt{I_2}\!
\end{bmatrix}
\pmb{R}(\theta)$, $\pmb{y} = [y_1\;y_2]^\mathsf{T}$, and $\pmb{n}_q =
[n_{q,1}\;n_{q,2}]^\mathsf{T}$ is the quantum noise vector with covariance
$\mathbb{E}[\pmb{n}_q\pmb{n}_q^T] = \sigma_q^2 \pmb{I}_2$.

\vspace{3pt}
\noindent {\it ML detection:}
Given the received vector $\pmb{y}$ and the instantaneous channel
coefficients $(I_1,I_2)$, the receiver performs joint detection over the two
time slots by solving the ML decision rule
\begin{equation}
	(\hat{x}_1,\hat{x}_2)
	= \arg\min_{(x_1,x_2)\in\{\pm\alpha\}^2}
	\left\|
	\pmb{y}
	-\pmb{H}\pmb{x}
	\right\|^{2}.
	\label{eq:ml_rule}
\end{equation}
This joint ML detector processes both time slots simultaneously, exploiting
the independent fading realizations $(I_1,I_2)$ to achieve a diversity order
of two. For $\theta=0$, the expression in \eqref{eq:ml_rule} reduces to conventional symbol-by-symbol detection, while non-zero $\theta$ enables cross-slot combination and diversity gain. Fig.~\ref{fig0} schematically illustrates the system model and the main building blocks of the proposed QRD scheme.

\noindent{\it Remarks:} In this work, we assume perfect channel state information (CSI) at the receiver; accordingly, the ML detector in \eqref{eq:ml_rule} assumes access to the instantaneous irradiance coefficients of the Gamma--Gamma channel. Since these coefficients are classical channel parameters, they can be estimated using conventional pilot-assisted channel estimation algorithms ({\it e.g.}, maximum-likelihood, least-squares, or linear minimum mean-square-error estimation) \cite[Ch.~10]{HAM}, without conflicting with quantum measurement constraints. The investigation of imperfect CSI scenarios or non-coherent detection schemes is left for future work.

It is also worth noting that homodyne detection does not implement the Helstrom-optimal positive operator-valued measure (POVM) \cite{DJO} for displaced squeezed states, and therefore the resulting receiver is suboptimal from a quantum-detection perspective. In this context, squeezing improves performance by reducing the variance of the measured quadrature, which enhances the statistics of the homodyne observations and leads to improved ML detection performance.

As for the complexity of the proposed QRD scheme, the additional implementation overhead is limited to a fixed $2\times2$ orthogonal rotation at the transmitter and joint ML detection over a small, constant number of hypotheses at the receiver (four for BPSK), resulting in a negligible and bounded increase in complexity compared to the baseline no-rotation scheme.

\section{Error probability and asymptotic analysis}

In this section, we study the error probability performance of the proposed QRD scheme and quantify the resulting diversity/coding gains.

\subsection{PEP and SER performance}

We focus on the pairwise error probability (PEP), defined as the probability that the transmitter emits the codeword $\pmb{x}$ while the receiver decides in favour of an alternative codeword $\tilde{\pmb{x}}$. The conditional PEP (given the channel coefficients $I_1,I_2$) is written as \cite{TSE}
\begin{align}
	\mathrm{PEP}(\pmb{x}\!\to\!\tilde{\pmb{x}}\mid I_1,I_2)
	&= \Pr\left\{
	\left\|\pmb{y}-\pmb{H}\tilde{\pmb{x}}
	\right\|^{2} < \left\|
	\pmb{y}
	-\pmb{H}\pmb{x}
	\right\|^{2}
	\right\}, \nonumber \\
	&=Q\left(\sqrt{\frac{\|\pmb{H}(\pmb{x}-\tilde{\pmb{x}}) \|^2}{4\sigma_q^2}} \right).
	\label{eq:pep_def}
\end{align}
Let $\pmb{u}=\pmb{R}(\theta)(\pmb{x}-\tilde{\pmb{x}})=[u_1\;u_2]^\mathsf{T}$, then $\|\pmb{H}(\pmb{x}-\tilde{\pmb{x}}) \|^2=\eta(I_1u_1^2+I_2u_2^2)$. Therefore, \eqref{eq:pep_def} is written as
\begin{align}
	\mathrm{PEP}(\pmb{x}\!\to\!\tilde{\pmb{x}}\mid I_1,I_2)
	&= Q\!\left(\sqrt{\frac{\eta e^{2r}\big(I_1 u_1^2+I_2 u_2^2\big)}{2}}\right).
	\label{eq:pep_cond_simplified}
\end{align}
To compute the average PEP, we incorporate Craig's formula $Q(\sqrt{x})=\tfrac{1}{\pi}\int_0^{\pi/2}
\exp\!\big(-x/[2\sin^2\vartheta]\big)d\vartheta$ \cite{TSE} and we average over the independent Gamma--Gamma $I_1$, $I_2$ coefficients. We have
\begin{equation}
	\overline{\mathrm{PEP}}(\pmb{x}\rightarrow \tilde{\pmb{x}})=
	\frac{1}{\pi}\int_{0}^{\pi/2}
	\mathcal L_I\!\left(\frac{\eta e^{2r} u_1^2}{4\,\sin^2\vartheta}\right)
	\mathcal L_I\!\left(\frac{\eta e^{2r} u_2^2}{4\,\sin^2\vartheta}\right)
	\,d\vartheta,
	\label{eq:pep_avg}
\end{equation}
where $\mathcal{L}_I(s)=\mathbb{E}_I[e^{-sI}]$ denotes the Laplace transform of the Gamma--Gamma distributed random variable $I$. A closed-form expression for $\mathcal{L}_I(s)$ is provided in Appendix A, derived in terms of the confluent hypergeometric function.

By using the standard union bound on the ML symbol error probability, which upper-bounds the error probability by summing all pairwise error events, the SER can be approximated by the sum of the PEPs \cite[Sec. 3.2.2]{TSE}. Due to the inherent symmetry of the considered code, two distinct PEP values can be identified, corresponding to codeword pairs that differ in one or in two symbol positions. If $\overline{\mathrm{PEP}}_1$ and $\overline{\mathrm{PEP}}_2$ denote the average PEPs for codeword pairs differing in one and two positions, respectively, the SER can be approximated as
\begin{equation}
	P_s \approx 2\overline{\mathrm{PEP}}_1 + \overline{\mathrm{PEP}}_2,
\end{equation}
where the unequal weighting arises because two neighbouring codeword pairs differ in a single symbol position, while only one pair differs in two positions.

To investigate the impact of the design parameters (the displacement amplitude $\alpha$, the squeezing parameter $r$, and the rotation angle $\theta$) on the overall system performance, we introduce the following optimization problem:
\begin{align}
	[P1]\;\; &(\alpha^*, r^*, \theta^*) = \arg\min_{\alpha,r,\theta} P_s \nonumber \\
	&\text{subject to} \quad \alpha^2 + \sinh^2(r) \le N~.
\end{align}
Although a closed-form expression for $P_s$ is generally analytically intractable, the optimal design can be inferred from the fundamental properties of the Laplace transform, as established in the following Proposition.

\begin{proposition}\label{propo}
	For the proposed QRD BPSK over i.i.d Gamma--Gamma fading with shapes $\epsilon,\zeta$ and $\epsilon\neq \zeta$, the optimal design that minimizes SER is attained at $\theta^\star=\frac{1}{2}\arctan(2)$, $r^\star=\tfrac{1}{2}\ln(2N+1)$, and $\alpha^\star=\sqrt{\tfrac{N(N+1)}{2N+1}}$. For the optimal configuration $P_s\approx 3\overline{PEP}_1$ with $\overline{PEP}_1=\overline{PEP}_2$.
\end{proposition}
\begin{proof}
	The proof is provided in Appendix~\ref{apB}.
\end{proof}

\setcounter{theorem}{0}
\subsection{Asymptotic analysis}
For high signal-to-noise ratios (SNRs)\footnote{In photon-limited quantum optical links with homodyne detection, the total mean photon number per symbol $N$ represents the available energy budget and serves as a standard proxy for the effective SNR.}, \textit{i.e.}, as $N \!\to\! \infty$, the $P_s$ expression admits a simplified asymptotic form that enables the analytical evaluation of fundamental communication-theoretic metrics such as the \emph{diversity gain} and the \emph{coding gain}\footnote{In the asymptotic expression $\mathrm{SER} = C\,\mathrm{SNR}^{-d}$, the parameter $d$ denotes the diversity gain (slope of the error curve in a log--log scale), while the constant $C$ determines the horizontal shift of the curve and thus corresponds to the coding gain, defined as $G_c = C^{-1/d}$ \cite{TSE}.}. The following theorem presents the diversity and coding gain for the proposed QRD scheme.
\begin{theorem}[High-SNR diversity and maximum coding gain]\label{theo1}
	For the QRD BPSK over i.i.d.\ Gamma--Gamma fading with shapes $(\epsilon,\zeta)$ and $\epsilon\neq \zeta$, let
	$g=\min\{\epsilon,\zeta\}$ and define the scheme-dependent constant as
	\begin{align}
		C_1&=3\frac{2^{6g-1}(\eta \sin(2\theta^*))^{-2g}}{\pi}\,\frac{\Gamma(|\epsilon-\zeta|)^2\Gamma(g)^2}{\Gamma(\epsilon)^2\Gamma(\zeta)^2}\,(\epsilon\zeta)^{2g} \nonumber \\
		&\;\;\;\times B\Big(2g+\tfrac12,\,2g+\tfrac12\Big).
	\end{align}
	The diversity and coding gains (as $N\to\infty$) are given by 
	
	\textbf{Diversity order:}
	\begin{align}
		\boxed{\,d = -\lim_{N\rightarrow \infty} \frac{\log P_s(N)}{\log N}
			=\begin{cases}
				4g, & \text{for}\; r> 0,\\
				2g, & \text{for}\; r= 0
		\end{cases}}~.
	\end{align}
	
	\textbf{Maximum coding gain:}
	\begin{align}
		\boxed{G_{c,\max}=C_1^{-1/{d}}}~.
	\end{align}
\end{theorem}
\begin{proof}
	Appendix~\ref{apC} derives the asymptotic forms of the Gamma-Gamma PDF and its Laplace transform, while Appendix~\ref{apD} develops the framework used to determine the diversity order and the (maximal) coding gain.
\end{proof}
The analysis shows that the diversity and coding gains reach their 
maximum values, $d = 4g$ and $G_{c,\max} = C_{1}^{-1/d}$, when the rotation angle and energy allocation are jointly optimized. This occurs for  $\theta^{\star} = \tfrac{1}{2}\arctan(2)$ and an (almost) equal distribution of photons  between displacement and squeezing, revealing that squeezing not only enhances robustness but effectively doubles the achievable diversity.

It is worth noting that squeezing reduces the variance of the measured quadrature and, under the total photon-number constraint, modifies the scaling of the effective decision statistic $e^{2r}\alpha^2$ with $N$. When both displacement and squeezing scale with $N$, the resulting symbol-error probability exhibits an additional $N^{-2g}$ decay factor, yielding an asymptotic exponent $d=4g$. This effect does not correspond to additional fading paths but rather to an improved SNR exponent, which can be interpreted as an effective or super-diversity gain.

\subsection{Baseline (no rotation, single-slot)}

For comparison purposes, we present here the \textit{baseline scheme}, where no rotation is employed and communication is performed in a single time slot \cite{CHE}. By using the output of the homodyne detection in \eqref{eq:homodyne_output}, SER can be expressed as
\begin{align}
	P_s = \frac{1}{\pi} \int_{0}^{\frac{\pi}{2}} 
	\mathcal{L}_I\!\left( \frac{\eta e^{2r}\alpha^2}{\sin^2(\vartheta)} \right)
	d\vartheta,\label{asy2}
\end{align}
where $\mathcal{L}_I(\cdot)$ is given in Appendix~\ref{apA} (non-asymptotic) and Appendix~\ref{apC} (asymptotic). By applying similar arguments to Proposition~\ref{propo}, the configuration that minimizes the SER probability is identical to that of the QRD scheme.

For the asymptotic analysis, by using the high-SNR form of the Laplace transform and following similar algebraic steps as in the previous section, we find that the optimal configuration again splits the total energy symmetrically between displacement and squeezing. The corresponding diversity and maximum coding gains are given by
\begin{align}
	&d =
	\begin{cases}
		2g, & \text{for } r>0,\\[2pt]
		g, & \text{for } r=0,
	\end{cases}\\[4pt]
	&G_{c,\max} = C_b^{-1/d}, \label{asym4}
\end{align}
where $C_b=\frac{2^{2g-1} \eta^{-g}}{\pi}
\,\frac{\Gamma(|\epsilon-\zeta|)\Gamma(g)}{\Gamma(\epsilon)\Gamma(\zeta)}\,
(\epsilon\zeta)^{g}
B\!\left(g+\tfrac{1}{2},\,g+\tfrac{1}{2}\right)$. Therefore, we observe that the optimal configuration coincides with that of the proposed QRD scheme. Squeezing remains beneficial for the baseline case as well; however, both the diversity and the coding gains are reduced, since only a single fading realization contributes to the decision process. In particular, the baseline scheme achieves half the diversity order and a lower coding gain compared to the rotation-based design.

\begin{figure}
    \centering
    \includegraphics[width=\linewidth]{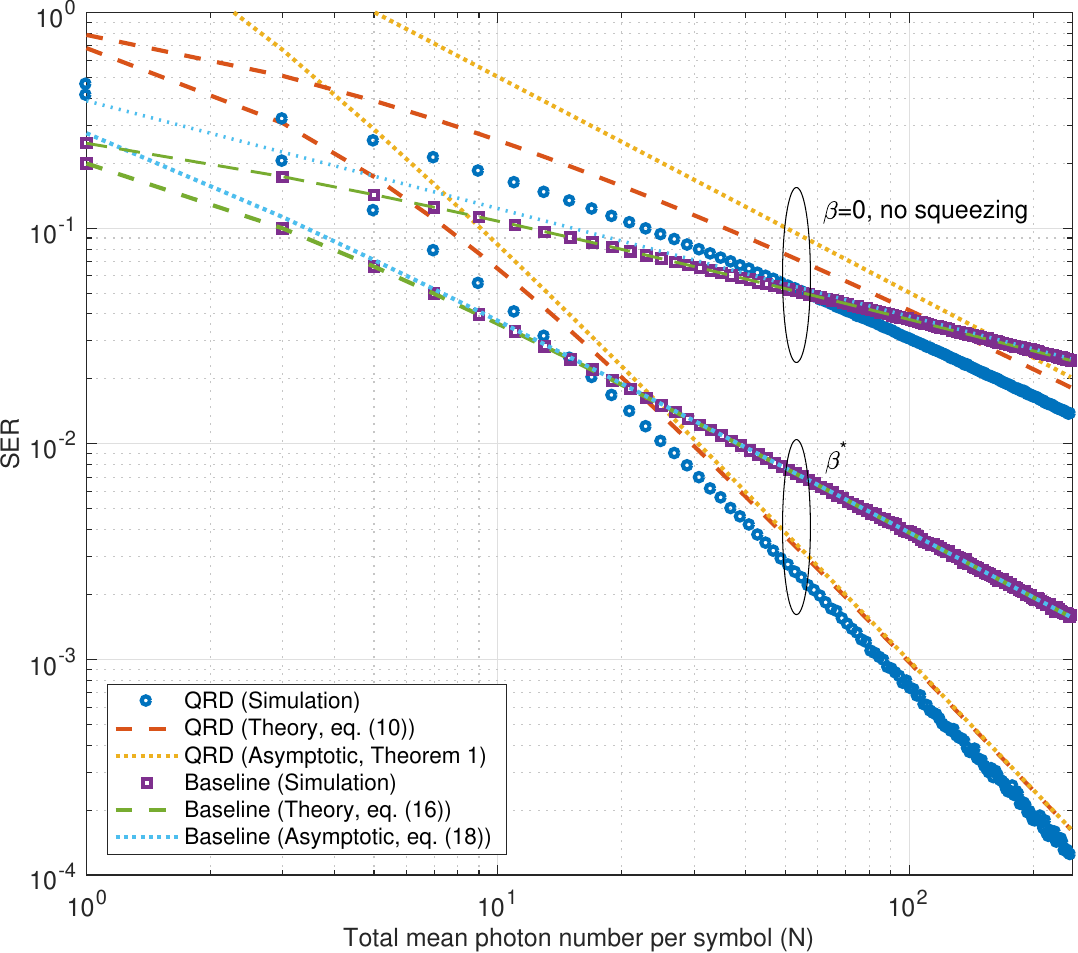}
    \caption{SER performance of QRD and baseline BPSK schemes versus the total mean photon number per symbol $N$; QRD uses an optimal angle $\theta^*=\frac{1}{2}\arctan(2)$.}
    \label{fig1}
\end{figure}

\begin{figure}
\centering
\includegraphics[width=\linewidth]{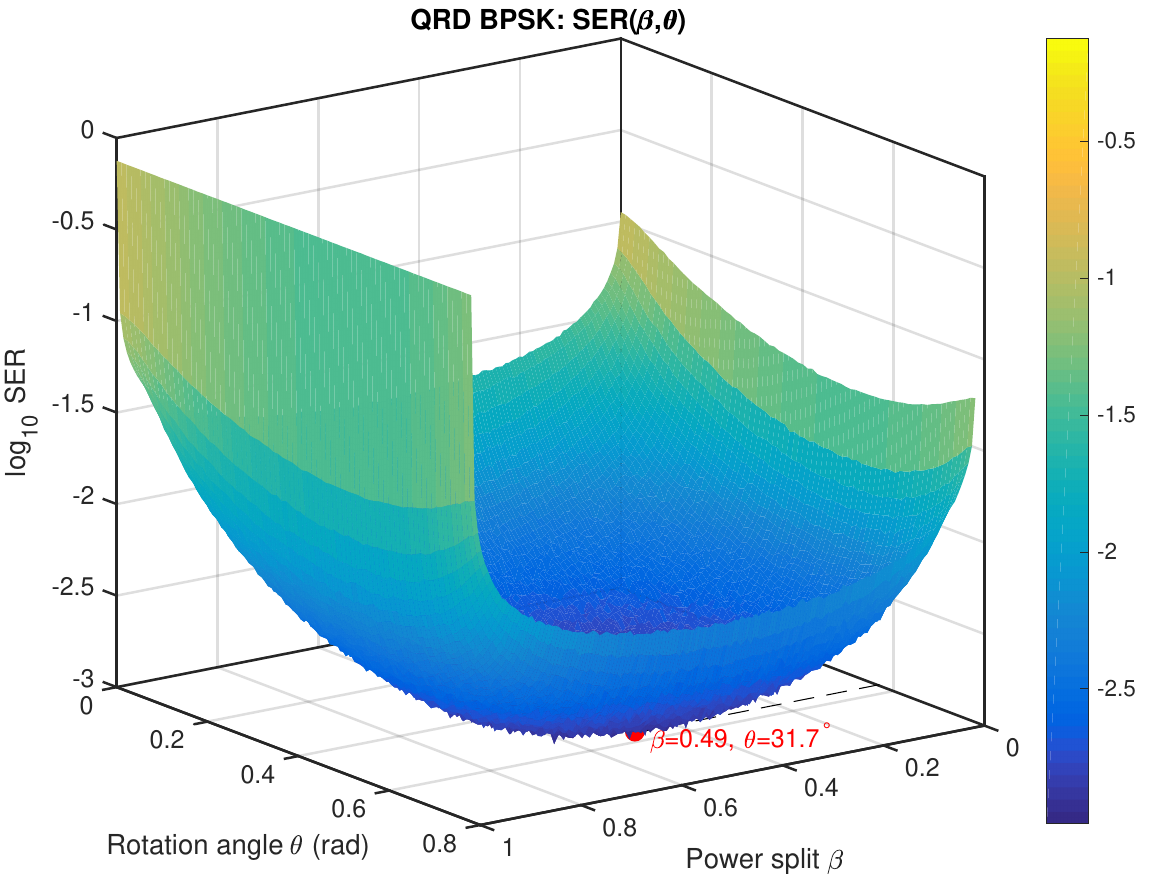}
\caption{SER performance of the proposed QRD scheme versus the rotation angle $\theta$ and the power split $\beta$; $N=80$ photons per symbol.}
\label{fig2}
\end{figure}

\section{Numerical studies}
Computer simulations are carried out to validate the efficiency of the proposed QRD scheme. The simulation setup assumes $\eta = 0.8$, $\epsilon = 0.5$, and $\zeta = 1.2$, corresponding to strong turbulence conditions. 

Fig.~\ref{fig1} illustrates the SER performance versus the total mean number of photons per symbol $N$ for both the QRD scheme and the baseline (single-slot) scheme; an optimal angle $\theta^*=0.5\arctan(2)$ is used for the QRD scheme. For fairness, both schemes are compared under the same total photon (energy) budget per symbol. Two configurations are considered {\it i.e.,} the optimal $\beta^*$ and $\beta = 0$ (no squeezing). The first key observation is that the proposed QRD scheme outperforms the baseline at high values of $N$ ({\it i.e.,} high SNRs) and achieves a diversity order twice that of the baseline. Specifically, the optimal $\beta^*$ yields a diversity order of $4g$, while the baseline scheme is limited to $2g$. When squeezing is not employed ($\beta = 0$), both the rotation-based and baseline schemes lose half of their diversity, resulting in $2g$ and $g$, respectively. This confirms the beneficial role of squeezing in Gamma--Gamma fading channels, as it effectively provides an additional diversity branch for both schemes. It can be also observed that for a low number of photons ({\it i.e.,} low SNRs $N<20$), the system performance is dominated by quantum noise, rendering the QRD scheme less efficient.
Finally, the theoretical results derived from the analytical expressions in~\eqref{eq:pep_avg} and~\eqref{asy2}, as well as the asymptotic formulations in Theorem~\ref{theo1} and~\eqref{asym4}, closely match the simulation outcomes, thereby validating the proposed mathematical framework.

Fig.~\ref{fig2} plots the SER of the QRD scheme versus the rotation angle $\theta$ and the power split $\beta$. The setup assumes $N=80$ photons per symbol ({\it i.e.,} $40$ photons per slot). We can observe the critical impact of both parameters on the SER performance, which exhibits a convex behaviour. The minimum value is achieved near $\beta^* \approx N/(2N+1) = 0.49$ and $\theta^* \approx 0.5\arctan(2) = 31.72^\circ$ which validate our theoretical results in Proposition \ref{propo}.

\section{Conclusion}

This letter introduces, for the first time in the literature, the concept of QRD, which achieves diversity gains over Gamma--Gamma fading channels without any loss in spectral efficiency. Analytical and simulation results consistently demonstrate the performance advantages of the proposed scheme and reveal the super-diversity effects induced by the squeezing process. The proposed QRD framework represents a promising solution for quantum optical communication systems with stringent reliability requirements, while future work may extend the approach to more comprehensive FSO channel models and alternative receiver or CSI-limited scenarios.

\appendices 
\section{Laplace transform of the Gamma--Gamma distribution}\label{apA}

The Laplace transform of a Gamma--Gamma random variable is given by 
\begin{align}
	\mathcal L_I(s)=
	\frac{2(\epsilon\zeta)^{\frac{\epsilon+\zeta}{2}}}
	{\Gamma(\epsilon)\Gamma(\zeta)}
	\int_{0}^{\infty}
	e^{-sz}\,z^{\frac{\epsilon+\zeta}{2}-1}
	K_{\epsilon-\zeta}\big(2\sqrt{\epsilon\zeta z}\big) dz.\label{int1}
\end{align}
By combining the expressions in \cite[6.643.3]{GRA} and \cite[9.220.4]{GRA} with
$\mu=\tfrac{\epsilon+\zeta-1}{2}$,
$\nu=\tfrac{\epsilon-\zeta}{2}$,
$\alpha=s$, and $b=\sqrt{\epsilon\zeta}$, the expression in \eqref{int1} becomes
\begin{align}
	\mathcal L_I(s)=
	\Big(\frac{\epsilon\zeta}{s}\Big)^{\epsilon}
	\Psi \left(\epsilon,\epsilon+1-\zeta;\frac{\epsilon\zeta}{s} \right),
\end{align}
where $\Psi(\cdot,\cdot;\cdot)$ denotes the confluent hypergeometric function of the second kind (Tricomi function) \cite[9.210]{GRA}.

\section{Minimize PEP based on Laplace transformation properties}\label{apB}

Let $I\!\ge0$ with Laplace transform $\mathcal{L}_I(s)=\mathbb{E}[e^{-sI}]$, $s>0$. From the basic properties of the Laplace transformation, $\mathcal{L}_I$ is strictly decreasing and \emph{log-convex} in $s$. For the two-slot BPSK code and rotation $\pmb{R}(\theta)$, let $\pmb{u}=\pmb{R}(\theta)(\pmb{x}-\tilde{\pmb{x}})=[u_1\,u_2]^T$.
Due to symmetry, it is sufficient to study one representative pair from each PEP class; for codeword pairs that differ in one symbol position $(u_1,u_2)=(2\alpha\cos\theta,2\alpha\sin\theta)$, and for codeword pairs that differ in two symbol positions $(u_1,u_2)=\big(2\alpha(\cos\theta-\sin\theta),2\alpha(\sin\theta+\cos\theta)\big)$. Since $\pmb{R}(\theta)$ is orthonormal, we note that $u_1^2+u_2^2=\|\pmb{x}-\tilde{\pmb{x}}\|^2=4\alpha^2$.
Using Craig's formula, the PEP integrand in \eqref{eq:pep_avg} at angle $\vartheta$ is
\begin{align}
	\mathcal{I}(\vartheta;\theta)=\mathcal{L}_I(s_1)\,\mathcal{L}_I(s_2), \qquad
	s_i=\frac{\eta e^{2r}u_i^2}{2\sin^2\vartheta},\ i=1,2,
\end{align}
so that for each fixed $\vartheta$, the sum
\begin{align}
	s_1+s_2=\frac{\eta e^{2r}}{2\sin^2\vartheta}(u_1^2+u_2^2)
	=\frac{2\eta e^{2r}\alpha^2}{\sin^2\vartheta}, \label{op2}
\end{align}
is independent of $\theta$. By log-convexity, for fixed $S=s_1+s_2$ we have $\mathcal{L}_I(s_1)\mathcal{L}_I(s_2)\ge[\mathcal{L}_I(S/2)]^2$,
with equality if and only if $s_1=s_2$. Therefore, to minimize the PEP, we need to minimize the imbalance $|s_1-s_2|$, which is equivalent to maximizing the product
\begin{align}
	s_1s_2=\frac{S^2-(s_1-s_2)^2}{4}.
\end{align}
For the two codeword-pair classes, we have $s_1s_2\propto u_1^2u_2^2 \propto \sin^2(2\theta)$ (one position), and $s_1s_2 \propto 4\cos^2(2\theta)$ (two positions). Therefore, to minimize the worst PEP across both pairs (min-max design), we enforce $\sin^2(2\theta)=4\cos^2(2\theta)\Rightarrow \theta^*=\tfrac{1}{2}\arctan(2)$.

As for the optimal power allocation, minimizing the average PEP is equivalent to maximizing the term $e^{2r}\alpha^2$ in \eqref{op2} subject to $\alpha^2+\sinh^2 r\le N$. To simplify the computation, we define the fraction of total energy allocated to squeezing $\beta=\frac{\sinh^2(r)}{N}$ and thus $e^{2r}\alpha^2=(\sqrt{\beta N}+\sqrt{\beta N+1})^2(1-\beta)N$, where $\sinh^2(r)=\beta N\Rightarrow \sinh(r)=\pm \sqrt{\beta N}\Rightarrow r=\pm \ln(\sqrt{\beta N}+\sqrt{\beta N+1})\Rightarrow r=\ln(\sqrt{\beta N}+\sqrt{\beta N+1})$ since $r\geq 0$. Therefore, the optimal power split corresponds to the optimization problem 
\begin{align}
	\beta^*=\arg\max_{\beta \in [0\;1]}(\sqrt{\beta N}+\sqrt{\beta N+1})^2(1-\beta).
\end{align} 
By the first-order optimality condition, we set the derivative to zero and solve for the stationary point, which yields
\begin{align} 
	\beta^*=\frac{N}{2N+1}~.
\end{align}
As $N\rightarrow \infty$, $\beta^*\rightarrow 1/2$ which is consistent with the asymptotic maximization in Appendix \ref{apD}.

\section{Asymptotic expressions for high SNRs}\label{apC}

For $z\rightarrow 0$,  the modified Bessel function is simplified to $K_{\epsilon-\zeta}(z)
\sim
\frac{1}{2}\Gamma(|\epsilon-\zeta|)
\left(\frac{z}{2}\right)^{-|\epsilon-\zeta|}$ \cite{GRA}. Therefore, the Gamma-Gamma PDF is simplified to 
\begin{align}
	f_I(z)\sim 
	\frac{\Gamma(|\epsilon-\zeta|)}{\Gamma(\epsilon)\Gamma(\zeta)}\;
	(\epsilon\zeta)^{g}\;
	z^{g-1},\;\;\text{for}\; \epsilon\neq \zeta
\end{align}
where $g=\min\{\epsilon,\zeta\}$. 
By using the above asymptotic expressions, the (asymptotic with $s\rightarrow \infty$) Laplace transformation of the Gamma-Gamma distribution simplifies to 
\begin{align}
	\mathcal L_I(s)&\sim
	\frac{\Gamma(|\epsilon-\zeta|)\,\Gamma(g)}{\Gamma(\epsilon)\Gamma(\zeta)}
	(\epsilon\zeta)^{g}s^{-g},\;\textrm{for}\;\epsilon\neq \zeta, \nonumber \\
	&\triangleq \Lambda s^{-g},\label{w1}
\end{align}
where \eqref{w1} uses the expression in \cite[3.351.3]{GRA}. By using the above asymptotic expression, the average PEP in ~\eqref{eq:pep_avg} is simplified as follows 
\begin{align}
	&\overline{\mathrm{PEP}}\sim \frac{2^{4g} (u_1^2u_2^2)^{-g} \Lambda^2}{\pi}(\eta\,e^{2r})^{-2g}
	\int_{0}^{\pi/2}\!\big(\sin^2\vartheta\big)^{2g}\,d\vartheta \nonumber \\
	&=\frac{2^{8g-1} \Lambda^2 (u_1'^2 u_2'^2)^{-g}\eta^{-2g}}{\pi}B\!\left(2g+\frac{1}{2},2g+\frac{1}{2} \right)\! \big(e^{2r}\alpha^2\big)^{-2g}, \label{div0}\\
	&=\frac{2^{6g-1} \Lambda^2 (\eta \sin(2\theta^\star))^{-2g}}{\pi}B\!\left(2g+\frac{1}{2},2g+\frac{1}{2} \right) \nonumber  \\
	&\;\;\;\;\;\times \big(e^{2r}\alpha^2\big)^{-2g}\;\;\triangleq C_0 \big(e^{2r}\alpha^2\big)^{-2g}, \label{div1}
\end{align}
where the above integration relies on \cite[3.621.3]{GRA}, and $u_i'=u_i/\alpha$ (normalized values). From \eqref{div0}, we observe that minimizing the PEP is equivalent to maximizing the product $u_1'^2 u_2'^2$, which is exactly the same criterion as in the non-asymptotic case (see Appendix~\ref{apB}). By the same argument used there, the optimal rotation is $\theta^\star=\tfrac{1}{2}\arctan(2)$, which minimizes the error probability in the asymptotic regime as well; consequently, $u_1'^2 u_2'^2=4 \sin^2\!\big(2\theta^\star\big)$.

\section{Diversity and coding gain}\label{apD}

For the optimal angle configuration, we have $\overline{\mathrm{PEP}}_1=\overline{\mathrm{PEP}}_2=\overline{\mathrm{PEP}}$ and thus $P_s\approx 3\overline{\mathrm{PEP}}= 3C_0(e^{2r}\alpha^2)^{-2g}\triangleq C_1(e^{2r}\alpha^2)^{-2g}$.
Given the channel energy $N$, we define the fraction of squeezing $\beta=\frac{\sinh^2(r)}{N}$. In this case, by using the expression in \eqref{div1}, $P_s$ can be written as 
\begin{align}
	P_s&\sim C_1\big((\sqrt{\beta N}+\sqrt{\beta N+1})^2(1-\beta)N\big)^{-2g},  \nonumber \\
	&\sim \begin{cases}
		C_1\big(4\beta (1-\beta)\big)^{-2g}N^{-4g}, & \text{for } \beta > 0,\\
		C_1 N^{-2g}, & \text{for } \beta = 0~.
	\end{cases}
\end{align}
Therefore, the diversity gain is given by 
\begin{align}
	d &= -\lim_{N\rightarrow \infty} \frac{\log P_s(N)}{\log N}
	=\begin{cases}
		4g, & \text{for } \beta > 0,\\
		2g, & \text{for } \beta = 0~.
	\end{cases}
\end{align}
So we can see that squeezing affects diversity gain, and $\beta>0$ ensures the maximum achievable diversity gain equal to $d=4g$, and a coding gain $G_c=C_1^{-1/d}\big(4\beta (1-\beta)\big)^{2g/d}$. Therefore, the optimal energy allocation (that maximizes both the diversity and the coding gains) relays to the following optimization problem 
\begin{align}
	&\max_{\beta}(1-\beta)\beta\;\;\;\textrm{subject to}\;0<\beta<1,
\end{align}
which gives $\beta^*=\frac{1}{2}$ and therefore $G_{c,\max}=C_1^{-1/d}$.

\bibliographystyle{IEEEtran}
\bibliography{references}

\begin{IEEEbiography}[{\includegraphics[width=1in,height=1.25in,clip,keepaspectratio]{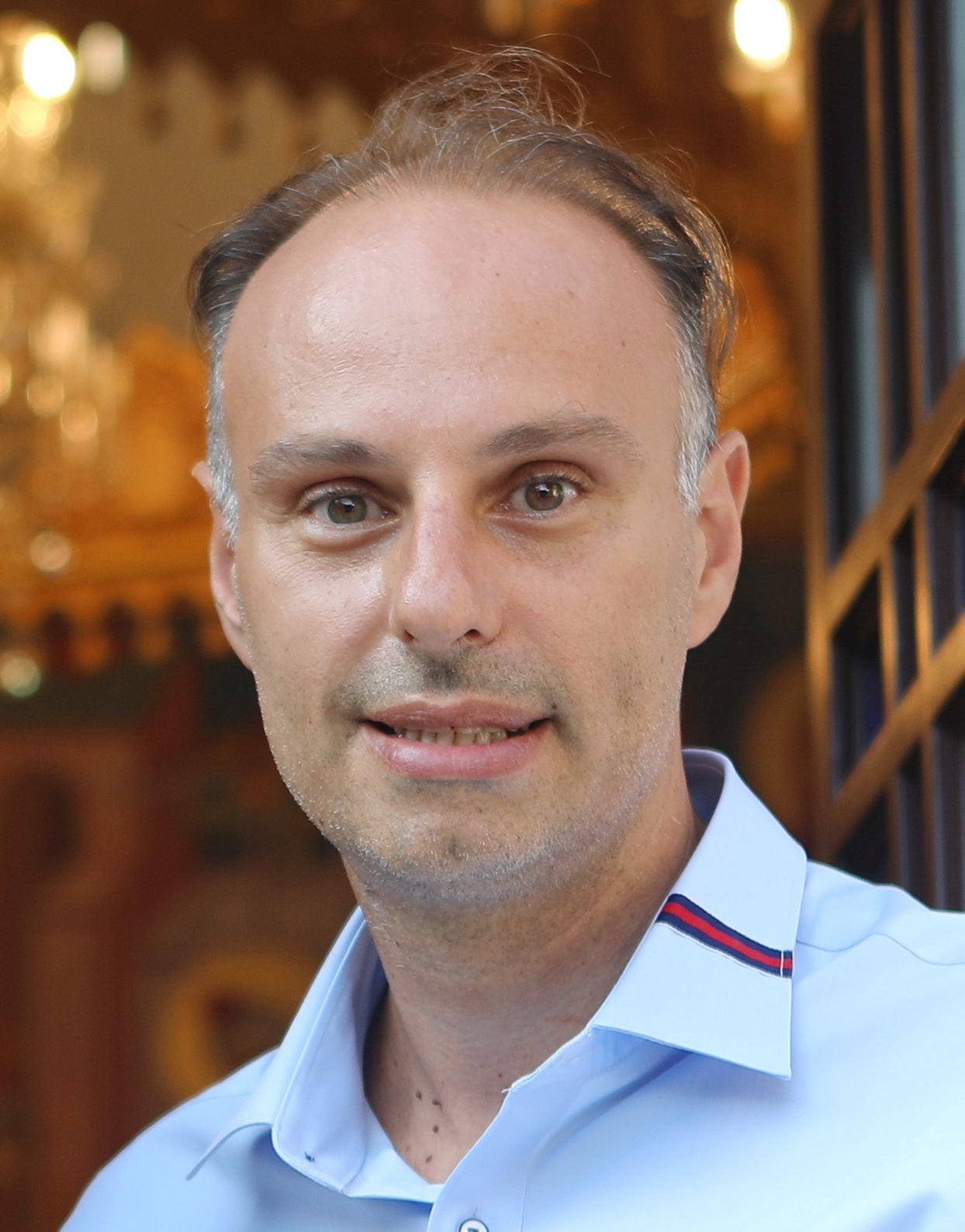}}]{Ioannis Krikidis} (F?19) received the diploma in Computer Engineering from the Computer Engineering and Informatics Department (CEID) of the University of Patras, Greece, in 2000, and the M.Sc and Ph.D degrees from \'Ecole Nationale Sup\'erieure des T\'el\'ecommunications (ENST), Paris, France, in 2001 and 2005, respectively, all in Electrical Engineering. From 2006 to 2007 he worked, as a Post-Doctoral researcher, with ENST, Paris, France, and from 2007 to 2010 he was a Research Fellow in the School of Engineering and Electronics at the University of Edinburgh, Edinburgh, UK.
	
He is currently a Professor at the Department of Electrical and Computer Engineering, University of Cyprus, Nicosia, Cyprus. His current research interests include wireless communications, quantum computing, 6G communication systems, wireless powered communications, and intelligent reflecting surfaces. Dr. Krikidis serves as an Associate Editor for IEEE Transactions on Wireless Communications, and Editor in Chief for Frontiers in Communications and Networks. He was the recipient of the Young Researcher Award from the Research Promotion Foundation, Cyprus, in 2013, and the recipient of the IEEEComSoc Best Young Professional Award in Academia, 2016, and IEEE Signal Processing Letters best paper award 2019. He has been recognized by the Web of Science as a Highly Cited Researcher for 2017-2021. He has received the prestigious ERC Consolidator Grant for his work on wireless powered communications.
\end{IEEEbiography}

\EOD

\end{document}